\documentclass[twoside,twocolumn,english]{article}
\usepackage[T1]{fontenc}
\usepackage[latin9]{inputenc}
\usepackage{amsmath}
\usepackage{amsthm}
\usepackage{amssymb}
\usepackage{graphicx}

\makeatletter
\theoremstyle{plain}
\newtheorem{thm}{\protect\theoremname}
\theoremstyle{definition}
\newtheorem{defn}[thm]{\protect\definitionname}
\theoremstyle{plain}
\newtheorem{lem}[thm]{\protect\lemmaname}
\ifx\proof\undefined
\newenvironment{proof}[1][\protect\proofname]{\par
	\normalfont\topsep6\p@\@plus6\p@\relax
	\trivlist
	\itemindent\parindent
	\item[\hskip\labelsep\scshape #1]\ignorespaces
}{%
	\endtrivlist\@endpefalse
}
\providecommand{\proofname}{Proof}
\fi
\theoremstyle{definition}
\newtheorem{example}[thm]{\protect\examplename}

\usepackage[letterpaper]{geometry}

\usepackage{balance}
\usepackage{url}
\usepackage{tikz}

\makeatother

\usepackage{babel}
\providecommand{\definitionname}{Definition}
\providecommand{\examplename}{Example}
\providecommand{\lemmaname}{Lemma}
\providecommand{\theoremname}{Theorem}

\begin{document}

\title{Geometric Sparsification of Closeness Relations:\\
Eigenvalue Clustering for Computing Matrix Functions}

\author{Nir Goren, Dan Halperin, and Sivan Toledo\\
Blavatnik School of Computer Science, Tel-Aviv University}
\maketitle
\begin{abstract}
\small\baselineskip=9ptWe show how to efficiently solve a clustering
problem that arises in a method to evaluate functions of matrices.
The problem requires finding the connected components of a graph whose
vertices are eigenvalues of a real or complex matrix and whose edges
are pairs of eigenvalues that are at most $\delta$ away from each
other. Davies and Higham proposed solving this problem by enumerating
the edges of the graph, which requires at least $\Omega(n^{2})$ work.
We show that the problem can be solved by computing the Delaunay triangulation
of the eigenvalues, removing from it long edges, and computing the
connected components of the remaining edges in the triangulation.
This leads to an $O(n\log n)$ algorithm. We have implemented both
algorithms using CGAL, a mature and sophisticated computational-geometry
software library, and we demonstrate that the new algorithm is much
faster in practice than the naive algorithm. We also present a tight
analysis of the naive algorithm, showing that it performs $\Theta(n^{2})$
work, and correct a misrepresentation in the original statement of
the problem. To the best of our knowledge, this is the first application
of computational geometry to solve a real-world problem in numerical
linear algebra. 
\end{abstract}

\section{Introduction}

This
paper proposes and analyzes efficient algorithms to sparsify transitive
closeness relations of points in the Euclidean plane. The problem
that we solve is an important step in a general method to efficiently
compute functions of matrices. 

More specifically, given a set $\Lambda$ of $n$ points in the plane
(real or complex eigenvalues of a matrix, in the underlying problem),
we wish to compute the connected components of a graph $G(\Lambda,\delta)$
whose vertices are the $n$ points and whose edges connect pairs of
points that are within distance at most $\delta$ of each other, for
some real $\delta>0$. Points that are at most $\delta$ apart are
said to be \emph{close}, and in this problem closeness is transitive.
The connected components of $G(\Lambda,\delta)$ partition $\Lambda$
into disjoint minimal well-separated clusters. That is, points in
two different clusters are more than $\delta$ apart, and the clusters
cannot be reduced while maintaining this property.

This problem is an important step in a method proposed by Davies and
Higham~\cite{Davies:2003:SPA}\cite[Chapter~9]{HighamFoM} to compute
a function $f(A)$ of a square real or complex matrix $A$. We describe
the overall method and the role of the eigenvalue-clustering problem
in it in Section~\ref{sec:Background}. Here it suffices to say that
the eigenvalue-clustering problem allows the use of a divide an conquer
strategy while reducing the likelihood of numerical instability. Nearby
eigenvalues in separate clusters create an instability risk; this
is why we want the clusters to be well separated. Large clusters reduce
the effectiveness of the divide and conquer strategy, which is why
clusters should be as small as possible.

Solving the problem in $O(n^{2}\alpha(n))$ time\footnote{In this paper we use the term \emph{time} to refer to the number of
machine instructions, ignoring issues of parallelism, locality of
reference, and so on. When we measure actual running times, we state
that the measurement units is seconds.}, where $\alpha$ is the inverse Ackermann function, is easy. We start
with minimal but illegal singleton clusters, and then test each of
the $n(n+1)/n=O(n^{2})$ eigenvalues pairs for closeness. If they
are close and in different clusters, we merge their two clusters.
The overall time bound assumes that the data structure that represents
the disjoint sets supports membership queries and merge operations
in $O(\alpha(n))$ time each, amortized over the entire algorithm~\cite[Chapter~21]{CLRS2}. 

The main contributions of this paper are two algorithms that solve
this problem in $O(n\log n)$ time. One, presented in Section~\ref{sec:algorithm-real},
is very simple but is only applicable when all the eigenvalues are
real (all the points lie on the real axis). The other algorithm, which
is applicable to any set of points in the plane, is also fairly simple,
but uses a sophisticated building block from computational geometry,
namely the Delaunay triangulation. We present this algorithm in Section~\ref{subsec:alg-delaunay}.
The Delaunay triangulation is also a graph whose vertices are $\Lambda$,
but it is planar and therefore sparse, having only $O(n)$ edges.
It turns out that when edges longer than $\delta$ are removed from
a Delaunay triangulation of $\Lambda$, the remaining graph has exactly
the same connected components as $G(\Lambda,\delta)$. The Delaunay
triangulation can be constructed in $O(n\log n)$ time, giving as
an effective sparsification mechanism for $G(\Lambda,\delta)$. The
algorithm for the real case also constructs a Delaunay triangulation,
but in this case the triangulation is particularly simple. 

Algorithms in computational geometry, like the algorithms that construct
the Delaunay triangulation, can suffer catastrophic failures when
implemented using floating-point arithmetic. Therefore, we implemented
our algorithms using CGAL, a computational-geometry software library
that supports both floating-point arithmetic and several types of
exact arithmetic systems. This implementation is described in detail
in Section~\ref{sec:Implementation}.

Experimental results, presented in Section~\ref{sec:Experimental-Results},
demonstrate that the new algorithms outperform the naive algorithm
by large margins. The results also demonstrate that the extra cost
of exact arithmetic is usually insignificant, at least when using
an arithmetic system that does use floating-point arithmetic whenever
possible. 

Our paper contains two additional contributions. The first, presented
in Section~\ref{subsec:alg-union-find}, is an amortized analysis
of the naive algorithm coupled with a particularly simple data structure
to represent disjoint sets. The analysis shows that even with this
simple data structure, proposed by Davies and Higham (and used many
times in the literature in various variants), the total running time
of the naive algorithm is only $O(n^{2})$. 

The second is an observation, presented in Appendix A, that an alternative
definition of the required eigenvalue partition, proposed by Davies
and Higham is not equivalent to the connected components of $G(\Lambda,\delta)$
and is not particularly useful in the overall method for evaluating
$f(A)$. 

Let's get started.

\section{\label{sec:Background}Background}

\begin{figure*}
\begin{centering}
\begin{tikzpicture}
  \draw [fill=lightgray] (0.00,0) rectangle (0.7,0.7);
  \node at (0.35,0.35) {$A$};

  \draw [->] (0.81,0.35) -- ++(0.7,0);
  \node [above] at (1.16,0.35) {\footnotesize Schur};
  \node [below] at (1.16,0.35) {\footnotesize $\Theta(n^3)$};
 
  \draw [fill=lightgray] (1.62,0) rectangle (2.32,0.7);
  \node at (1.97,0.35) {$Q_S$};

  \draw [fill=lightgray, ultra thin ] (2.43,0.0)
        ++(0,0.7) -- ++(0.7,0.0) -- ++(0,-0.7) 
        -- ++(-0.1,0.0) -- ++(0.0,0.1)
        -- ++(-0.1,0.0) -- ++(0.0,0.1)
        -- ++(-0.1,0.0) -- ++(0.0,0.1)
        -- ++(-0.1,0.0) -- ++(0.0,0.1)
        -- ++(-0.1,0.0) -- ++(0.0,0.1)
        -- ++(-0.1,0.0) -- ++(0.0,0.1)
        -- ++(-0.1,0.0) -- cycle;
  \draw                  (2.43,0) rectangle (3.13,0.7);

  \node at (2.78,0.35) {$T$};

  \draw [->, rounded corners] (2.78,0.8) -- ++(0.0,0.35) -- ++(2.43,0.0) -- ++(0.0,-0.35);
  \node [above] at (3.995,1.15) {\footnotesize cluster};

  \draw [fill=lightgray] (3.24,0) rectangle (3.94,0.7);
  \node at (3.59,0.35) {$Q^*_S$};
 
  \draw [fill=lightgray, ultra thin ] (4.86,0.0)
        ++(0,0.7) -- ++(0.7,0.0) -- ++(0,-0.7) 
        -- ++(-0.1,0.0) -- ++(0.0,0.1)
        -- ++(-0.1,0.0) -- ++(0.0,0.1)
        -- ++(-0.1,0.0) -- ++(0.0,0.1)
        -- ++(-0.1,0.0) -- ++(0.0,0.1)
        -- ++(-0.1,0.0) -- ++(0.0,0.1)
        -- ++(-0.1,0.0) -- ++(0.0,0.1)
        -- ++(-0.1,0.0) -- cycle;
  \draw [fill=red,   red,   ultra thin] (4.86,0.0) ++(0,0.7)   rectangle ++(0.1,-0.1);
  \draw [fill=blue,  blue,  ultra thin] (4.86,0.0) ++(0.1,0.6) rectangle ++(0.1,-0.1); 
  \draw [fill=green, green, ultra thin] (4.86,0.0) ++(0.2,0.5) rectangle ++(0.1,-0.1); 
  \draw [fill=blue,  blue,  ultra thin] (4.86,0.0) ++(0.3,0.4) rectangle ++(0.1,-0.1); 
  \draw [fill=red,   red,   ultra thin] (4.86,0.0) ++(0.4,0.3) rectangle ++(0.1,-0.1); 
  \draw [fill=blue,  blue,  ultra thin] (4.86,0.0) ++(0.5,0.2) rectangle ++(0.1,-0.1); 
  \draw [fill=green, green, ultra thin] (4.86,0.0) ++(0.6,0.1) rectangle ++(0.1,-0.1); 
  \draw                                 (4.86,0.0) rectangle ++(0.7,0.7);
  \node [below] at (5.21,0) {$T$};

  \draw [->] (5.67,0.35) -- ++(0.7,0);
  \node [above] at (6.02,0.35) {\footnotesize reord};
  \node [below] at (6.02,0.35) {\footnotesize $\Theta(n^3)$};

  \draw [fill=lightgray] (6.48,0) rectangle (7.18,0.7);
  \node at (6.83,0.35) {$Q_R$};

  \draw [fill=lightgray, ultra thin ] (7.29,0.0)
        ++(0,0.7) -- ++(0.7,0.0) -- ++(0,-0.7) 
        -- ++(-0.1,0.0) -- ++(0.0,0.1)
        -- ++(-0.1,0.0) -- ++(0.0,0.1)
        -- ++(-0.1,0.0) -- ++(0.0,0.1)
        -- ++(-0.1,0.0) -- ++(0.0,0.1)
        -- ++(-0.1,0.0) -- ++(0.0,0.1)
        -- ++(-0.1,0.0) -- ++(0.0,0.1)
        -- ++(-0.1,0.0) -- cycle;
  \draw [fill=red,   red,   ultra thin] (7.29,0.0) ++(0,0.7)   rectangle ++(0.1,-0.1);
  \draw [fill=red,   red,   ultra thin] (7.29,0.0) ++(0.1,0.6) rectangle ++(0.1,-0.1); 
  \draw [fill=green, green, ultra thin] (7.29,0.0) ++(0.2,0.5) rectangle ++(0.1,-0.1); 
  \draw [fill=green, green, ultra thin] (7.29,0.0) ++(0.3,0.4) rectangle ++(0.1,-0.1); 
  \draw [fill=blue,  blue,  ultra thin] (7.29,0.0) ++(0.4,0.3) rectangle ++(0.1,-0.1); 
  \draw [fill=blue,  blue,  ultra thin] (7.29,0.0) ++(0.5,0.2) rectangle ++(0.1,-0.1); 
  \draw [fill=blue,  blue,  ultra thin] (7.29,0.0) ++(0.6,0.1) rectangle ++(0.1,-0.1); 

  \draw [                   ultra thin] (7.29,0.0) ++(0.2,0.3)   rectangle ++(0,0.4);
  \draw [                   ultra thin] (7.29,0.0) ++(0.0,0.5)   rectangle ++(0.7,0);
  \draw [                   ultra thin] (7.29,0.0) ++(0.4,0.0)   rectangle ++(0,0.7);
  \draw [                   ultra thin] (7.29,0.0) ++(0.2,0.3)   rectangle ++(0.5,0);

  \draw                                 (7.29,0.0) rectangle ++(0.7,0.7);

  \node [below] at (7.64,0.0) {$T_R$};
  \draw [fill=lightgray] (8.10,0) rectangle (8.80,0.7);
  \node at (8.45,0.35) {$Q^*_R$};

  \draw [->, rounded corners] (7.64,0.8) -- ++(0.0,0.35) -- ++(2.43,0.0) -- ++(0.0,-0.35);
  \node [above] at (8.855,1.15) {\footnotesize block Parlett};

  \draw [fill=lightgray, ultra thin ] (9.72,0.0)
        ++(0,0.7) -- ++(0.7,0.0) -- ++(0,-0.7) 
        -- ++(-0.1,0.0) -- ++(0.0,0.1)
        -- ++(-0.1,0.0) -- ++(0.0,0.1)
        -- ++(-0.1,0.0) -- ++(0.0,0.1)
        -- ++(-0.1,0.0) -- ++(0.0,0.1)
        -- ++(-0.1,0.0) -- ++(0.0,0.1)
        -- ++(-0.1,0.0) -- ++(0.0,0.1)
        -- ++(-0.1,0.0) -- cycle;

  \draw [                   ultra thin] (9.72,0.0) ++(0.2,0.3)   rectangle ++(0,0.4);
  \draw [                   ultra thin] (9.72,0.0) ++(0.0,0.5)   rectangle ++(0.7,0);
  \draw [                   ultra thin] (9.72,0.0) ++(0.4,0.0)   rectangle ++(0,0.7);
  \draw [                   ultra thin] (9.72,0.0) ++(0.2,0.3)   rectangle ++(0.5,0);

  \draw            (9.72,0) rectangle (10.42,0.7);
  \node [below] at (10.07,0) {$f(T_R)$};

  \draw [->, rounded corners] (10.07,-0.55) -- ++(0.0,-0.35) -- ++(3.24,0.0) -- ++(0.0,0.35);

  \draw [fill=lightgray] (11.34,0) rectangle (12.04,0.7);
  \node at (11.69,0.35) {$Q_S$};
  \draw [fill=lightgray] (12.15,0) rectangle (12.85,0.7);
  \node at (12.50,0.35) {$Q_R$};

  \draw [fill=lightgray, ultra thin ] (12.96,0.0)
        ++(0,0.7) -- ++(0.7,0.0) -- ++(0,-0.7) 
        -- ++(-0.1,0.0) -- ++(0.0,0.1)
        -- ++(-0.1,0.0) -- ++(0.0,0.1)
        -- ++(-0.1,0.0) -- ++(0.0,0.1)
        -- ++(-0.1,0.0) -- ++(0.0,0.1)
        -- ++(-0.1,0.0) -- ++(0.0,0.1)
        -- ++(-0.1,0.0) -- ++(0.0,0.1)
        -- ++(-0.1,0.0) -- cycle;

  \draw            (12.96,0) rectangle (13.66,0.7);
  \node [below] at (13.31,0) {$f(T_R)$};
  \draw [fill=lightgray] (13.77,0) rectangle (14.47,0.7);
  \node at (14.12,0.35) {$Q^*_R$};
  \draw [fill=lightgray] (14.58,0) rectangle (15.28,0.7);
  \node at (14.93,0.35) {$Q^*_S$};

  \draw [->] (15.39,0.35) -- ++(0.7,0);
  \node [above] at (15.74,0.35) {\footnotesize mult};
  \node [below] at (15.74,0.35) {\footnotesize $\Theta(n^3)$};

  \draw [fill=lightgray] (16.20,0) rectangle (16.90,0.7);
  \node [below] at (16.55,0.0) {$f(A)$};

  \draw [->, rounded corners] (8.45,-0.1) -- ++(0.0,-1.0) -- ++(5.67,0.0)  -- ++(0.0,1.0);
  \draw [->, rounded corners] (6.83,-0.1) -- ++(0.0,-1.2) -- ++(5.67,0.0)  -- ++(0.0,1.2);
  \draw [->, rounded corners] (3.59,-0.1) -- ++(0.0,-1.4) -- ++(11.34,0.0) -- ++(0.0,1.4);
  \draw [->, rounded corners] (1.97,-0.1) -- ++(0.0,-1.6) -- ++(9.72,0.0) -- ++(0.0,1.6);

\end{tikzpicture}
\par\end{centering}
\caption{\label{fig:structure}The overall structure of the Davies-Higham method
for computing a function $f(A)$ of a matrix $A$.}
\end{figure*}
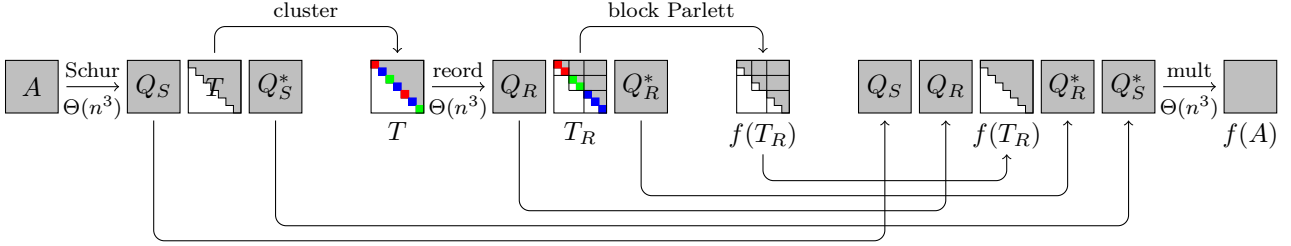
A scalar function $f\colon\mathbb{C}\rightarrow\mathbb{C}$ can be
extended to square real and complex matrices by letting $f$ act on
the eigenvalues of the matrix. That is, if $A\in\mathbb{C}^{n\times n}$
is diagonalizable so $A=VDV^{-1}$ with $D$ diagonal, then $f(A)=Vf(D)V^{-1}$
has the same eigenvectors as $A$ but eigenvalues that have been transformed
by $f$; here $f(D)$ denotes a diagonal matrix with diagonal entries
$(f(D)_{ii}=f(D_{ii})$. The definition can be extended to non-diagonalizable
matrices in one of several equivalent ways~\cite{HighamFoM}. Functions
of matrices have many applications~\cite{HighamFoM}.

For many functions $f$ of practical importance, such as the square
root and exponentiation ($f(x)=e^{x})$, there are specialized algorithms
to compute $f(A)$. There are also several general techniques to evaluate
$f(A)$. Among them is a sophisticated and efficient method due to
Davies and Higham~\cite{Davies:2003:SPA}\cite[Chapter~9]{HighamFoM}.
The problem that we solve in this paper is a subroutine the Davies-Higham
method.

The Davies-Higham method can be viewed in two ways. One is as a generalization
and adaptation of an older method due to Parlett~\cite{Parlett:1974:CFT}\cite[Section~4.6]{HighamFoM}.
The so-called \emph{Schur-Parlett} method computes the Schur decomposition
$A=Q_{S}TQ_{S}^{*}$, where $T$ is triangular and $Q_{s}$ is unitary,
evaluates $f(T)$ using a simple recurrence, and forms $f(A)=Q_{S}f(T)Q_{S}^{*}$.
This method is applicable to any function $f$, but it fails when
$A$ has repeated or highly clustered eigenvalues. When it does work,
this method evaluates $f(A)$ in $\Theta(n^{3})$ time. In particular,
all three steps of the method take cubic time time: the Schur decomposition,
the evaluation of $f(T)$, and the matrix multiplications required
to form $f(A)$ (the latter step can be asymptotically faster if one
uses fast matrix multiplication). The Davies-Higham method, which
is illustrated in Figure~\ref{fig:structure}, partitions the eigenvalues
into well-separated clusters, reorders the Schur decomposition $T=Q_{R}T_{R}Q_{R}^{*}$
so that clusters are contiguous along the diagonal of $T_{R}$, applies
some other algorithm to evaluate $f$ on diagonal blocks of $T_{R}$,
and then applies a block version of Parlett's recurrence to compute
the off-diagonal blocks of $T_{R}$. The partitioning of the spectrum
$\Lambda$ of $A$ into well-separated cluster is designed so that
the solution of the recurrence equations for the off-diagonal blocks
is numerically stable.

The other way to view the Davies-Higham is as a divide-and-conquer
algorithm. The technique that must be applied to evaluate $f$ on
diagonal blocks of $T_{R}$ have super-cubic cost. The technique that
Davies and Higham proposed is a Pade approximation of $f$, and its
cost is approximately quartic in the dimension of the block. Therefore,
it is best to apply this technique to diagonal blocks that are as
small as possible, to attain a total cost that is as close as possible
to cubic, not quartic. That is, the Davies-Higham chops the original
problem into sub-problems that are as small as possible (the diagonal
blocks of $T_{R}$), solves each one using an expensive algorithm,
and then merges the solutions. The splitting and merging phases are
cubic. 

Let us now review the entire Davies-Higham method, as illustrated
in Figure~\ref{fig:structure}. We start by computing the Schur decomposition
$A=Q_{S}TQ_{S}^{*}$. If $A$ is real with complex eigenvalues, we
compute the so-called \emph{real Schur decomposition}. In this case,
complex eigenvalues form conjugate pairs that are represented as $2$-by-$2$
diagonal blocks in $T$ (so $T$ is not triangular but block triangular
with $1$-by-$1$ and $2$-by-$2$ blocks). Next, we partition the
eigenvalues into clusters using a simple clustering rules described
below in Section~\ref{sec:The-Spectrum-Partitioning-Criter}. This
clustering algorithm is the main focus of this paper. In Figure~\ref{fig:structure},
the clusters are represented by coloring the eigenvalues, which lie
along the diagonal of $T$. Now we need to reorder the eigenvalues
so that clusters are contiguous while maintaining the triangular structure
and while maintaining the reordered matrix $T_{R}$ as a Schur factor
of $A$. That is, we transform $T$ into $T_{R}$ using unitary similarity.
The reordering also costs $O(n^{3})$ time~\cite{Bai:1993:SDB,KressnerBlockSchurReorderingTOMS}.
Now we evaluate $f$ on diagonal blocks of $T_{R}$ and then solve
Sylvester equations for the off-diagonal blocks of $f(T_{R})$. The
separation between clusters of eigenvalues is designed to minimize
errors in the solution of these equations. We note that the clustering
criterion proposed by Davies and Higham does not guarantee small errors;
it serves as a proxy for a criterion that is too difficult to use. 

When $A$ is a real matrix with complex eigenvalues, complex eigenvalues
form conjugate pairs and the two eigenvalues in each pair are kept
together in the reordering, in order to maintain the block-diagonal
structure of the Schur factor. We handle this case by including only
one eigenvalue from each pair in the input to the partitioning problem,
the one with positive imaginary part. Its conjugate is then placed
in the same cluster. 

\section{\label{sec:The-Spectrum-Partitioning-Criter}The Spectrum-Partitioning
Criterion}

Davies and Higham define the criteria for the partitionining of the
eigenvalues in two different ways. We present first the definition
that is both algorithmically useful and correct in the sense that
it serves the overall algorithm well.
\begin{defn}
The \emph{$\delta$-closeness graph} $G(\Lambda,\delta)$ of a set
of complex numbers $\Lambda=\{\lambda_{1},\lambda_{2},\ldots,\lambda_{n}\}$
(possibly with repetitions) is the graph whose vertex set is $\Lambda$
and whose edge set consists of all the pairs $\{\lambda_{i},\lambda_{j}\}$
for which $|\lambda_{i}-\lambda_{j}|\leq\delta$. 

We denote the connected components of $G=G(\Lambda,\delta)$ by $C_{1}^{(G)},\ldots,C_{k}^{(G)}$,
and when the graph is clear from the context, we denote the components
by $C_{1},\ldots,C_{k}$. We view connected components as sets of
vertices, so $C_{1},\ldots,C_{k}$ are disjoint sets of eigenvalues.
We denote the connected component in $G$ that contains $\lambda_{i}$
by $C^{(G)}(\lambda_{i})$ and by $C(\lambda_{i})$ if the graph is
clear from the context.
\end{defn}
Partitioning $\Lambda$ by connected components in $G(\Lambda,\delta)$
is effective in the Davies-Higham algorithm. This partitioning reduces
(in a heuristic sense explained in their paper) the risk of instability
while admitting efficient partitioning algorithms, including one proposed
in the Davies and Higham paper. We note that Davies and Higham imply
that the connected components of $G(\Lambda,\delta)$ are equivalent
to the a partition that satisfies two specific conditions, but this
is not the case, as we show in Appendix~A.

\section{\label{sec:algorithm-real}An Algorithm for Real Eigenvalues}

Davies and Higham proposed a partitioning algorithm that works for
both real and complex eigenvalues, but we start with a new algorithm
that is specialized for the real case and is both simpler and more
efficient than the Davies-Higham algorithm. We sort the eigenvalues
so that $\lambda_{\pi(1)}\leq\lambda_{\pi(2)}\leq\cdots\leq\lambda_{\pi(n)}$
($\pi$ is a permutation that sorts the eigenvalues). We then create
an integer vector $g$ of size $n$ and assign
\[
g_{i}=\begin{cases}
1 & |\lambda_{\pi(i)}-\lambda_{\pi(i-1)}|>\delta\\
0 & |\lambda_{\pi(i)}-\lambda_{\pi(i-1)}|\leq\delta\;,
\end{cases}
\]
denoting $\lambda_{\pi(0)}=-\infty$ so that $g_{1}$ is always $1$.
The vector $g$ marks gaps in the spectrum (the set of eigenvalues).
We now compute the prefix sums of $g$,
\[
c_{i}=\sum_{j=1}^{i}g_{i}\;.
\]
Now $c_{i}$ is the label (index) of the cluster that eigenvalue $\lambda_{\pi(i)}$
belongs to.

The running time of this technique is $\Theta(n\log n)$ assuming
that we use a comparison-based sorting algorithm. 

We defer the correctness proof for this algorithm to the next section,
because the proof is a special case of a more general analysis for
the complex case, but we state the result here.
\begin{thm}
\label{thm:correctness-of-sort-and-split}Partitioning $\Lambda\subset\mathbb{R}$
by sorting the eigenvalues and splitting whenever two adjacent eigenvalues
are more than $\delta$ away creates a partition that is identical
to the connected components of $G(\Lambda,\delta)$.
\end{thm}

\section{\label{sec:algorithm-complex}An Algorithm for Complex Eigenvalues}

If $A$ has complex eigenvalues, the simple method of Section~\ref{sec:algorithm-real}
no longer works. Later in this section we present a very efficient
algorithm to partition complex eigenvalues, but we start with a simpler
variant that is closer to the algorithm proposed by Davies and Higham.

\subsection{\label{subsec:alg-build-the-graph}The Davies-Higham Partitioning
Algorithm.}

Davies and Higham propose a partitioning algorithm that works for
both real and complex eigenvalues, but their paper (and Higham's book)
do not prove that it is correct, does not specify exactly how clusters
are represented, and does not analyze the complexity of the algorithm.

Their algorithm is incremental. It maintains a partitioning of a subset
of the eigenvalues. When step $t$ ends, the partitioning is valid
for the subgraph that contain all the vertices (eigenvalues) and all
the edges $\{\lambda_{i},\lambda_{j}\}$ for which $i\leq t$ or $j\leq t$.
Initially, every eigenvalue forms a singleton cluster, because we
have not considered any edges (closeness relations) yet. (The text
of Davies and Higham imples that the singleton cluster for $\lambda_{i}$
is formed only in the beginning of step $i$ and only if $\lambda_{i}$
is not already part of a larger cluster, but this only makes the algorithm
a little harder to understand.)

In step $i$, the algorithm computes the distances $|\lambda_{i}-\lambda_{j}|$
for all $j>i$ such that $\lambda_{j}$ is not already in the same
cluster as $\lambda_{i}$. If the distance is smaller than $\delta$,
meaning that a new edge has been discovered in the graph, the clusters
that contain $\lambda_{i}$ and $\lambda_{j}$ are merged.

Davies and Higham do not spell out exactly how clusters are represented,
but their text implies that they record in a vector $c$ the label
of the cluster of every eigenvalue; that is, $c_{i}$ is the label
(integer) of the cluster that contains $\lambda_{i}$. When they merge
clusters with indices $x$ and $y>x$, they relabel eigenvalues in
$y$ as belonging to $x$, and they decrease by $1$ every label higher
than $y$. The relabeling of clusters higher than $y$ may simplify
later phases in the overall algorithms, because at the end of the
algorithm the clusters are labeled contiguously $1,\ldots,k$, but
it is clearly also possible to relabel the clusters once at the end
in $\Theta(n)$ operations.

Davies and Higham do not prove the correctness of this algorithm (but
this is fairly trivial) and they do not analyze its complexity. The
loop structure of their algorithm shows that its running time is $\Omega(n^{2})$
and $O(n^{3})$ but the exact asymptotic complexity is not analyzed.

\subsection{\label{subsec:alg-union-find}Disjoint-Sets Data Structures for Connected
Components}

The Davies-Higham partitioning algorithm is an instantiation of a
generic method to compute connected components. The generic method
maintains a disjoint-sets data structure, initialized to a singleton
for every vertex. The edges of the graph are scanned, in any order.
For each edge $\{i,j\}$, the method determines the sets $S_{i}$
and $S_{j}$ that $i$ and $j$ belong to, respectively, and if $S_{i}\neq S_{j}$,
the two sets are merged. The correctness of the Davies-Higham algorithm
is a consequence of the correctness of this general method.

There are many ways to represent the sets and to perform the operations
that find $S_{i}$ given $i$ (the so-called \emph{find} operation)
and merge $S_{i}$ and $S_{j}$ (the so-called \emph{union }operation).
The most efficient general-purpose data structure uses rooted trees
to represent the sets and optimizations called \emph{union by rank}
and \emph{path compression} to speed up the operations; this data
structure and algorithms guarantee an $O(m\alpha(n))$ complexity
for a sequence of $m$ union or find operations on a set of $n$ elements
(in all the subsets combined), where $\alpha$ is the inverse Ackermann
function, whose value for any practical value of $n$ is at most $4$.
The Davies-Higham algorithm performs $n(n-1)/2=\Theta(n^{2})$ \emph{find
}operations and at most $n-1$ \emph{union} operations (since every
union operation reduces the number of subsets by $1$), so the complexity
with this data structure is $O(n^{2}\alpha(n))$. 

However, in our case even the simpler data structure and algorithms
that Davies and Higham proposed guarantee an $O(n^{2})$ complexity.
The number of union operations at most $n-1$, so even if every union
operation costs $\Theta(n)$ to scan the vector $c$ and to relabel
some of the components, the total cost of the union operations is
still $O(n^{2})$. The find operations cost $O(1)$, so their total
cost is again $O(n^{2})$.

\subsection{\label{subsec:alg-delaunay}An Efficient Geometric Partitioning Algorithm}

The $\delta$-closeness graph can have $n(n-1)/2=\Theta(n^{2})$ edges
so constructing the graph requires $\Theta(n^{2})$ operations. The
large number of edges also implies that the total cost of the disjoint-set
operations is high, $\Omega(n^{2})$. 

We have discovered that a sparse graph with only $O(n)$ edges and
that can be constructed in $O(n\log n)$ operations has exactly the
same connected components. This graph is the well-known \emph{Delaunay
triangulation} of the spectrum $\Lambda$, when viewed as a set of
points in the plane. We begin with definitions of the Delaunay triangulation
and of related geometric objects, specialized to the Euclidean plane,
as well with a statement of key properties of them and key relationships
between them. For further details on these objects, see~\cite{ComputationalGeometry2008,FortuneVoronoiDelaunay2018}. 
\begin{defn}
Given a set of points $\Lambda=\{\lambda_{1},\lambda_{2},\ldots,\lambda_{n}\}$
in the plane, the \emph{Voronoi cell} of $\lambda_{i}$ is the set
of all points that are closer to $\lambda_{i}$ than to any other
point in $\Lambda$. A \emph{Voronoi edge} is a nonempty set of points
that are equidistant from $\lambda_{i}$ and $\lambda_{j}\neq\lambda_{i}$
and closer to $\lambda_{i}$ and $\lambda_{j}$ than to any other
point in $\lambda$. A \emph{Voronoi vertex} is a point that is closest
to three or more points in $\Lambda$. The \emph{Voronoi diagram}
of $\Lambda$ is the ensemble of Voronoi faces, edges, and vertices. 
\end{defn}

\begin{defn}
The \emph{Delaunay Graph }of a set of points $\Lambda=\{\lambda_{1},\lambda_{2},\ldots,\lambda_{n}\}$
in the plane is the dual of their Voronoi diagram: $\{\lambda_{i},\lambda_{j}\}$
is an edge of the Delaunay graph if and only if the cells of $\lambda_{i}$
and $\lambda_{j}$ share an edge. 
\end{defn}
We note that if the Voronoi cells of $\lambda_{i}$ and $\lambda_{j}$
share a single point, then this point is a Voronoi vertex and not
a Voronoi edge, and in such a case $\{\lambda_{i},\lambda_{j}\}$
is \emph{not} an edge of the Delaunay graph. In many cases it is convenient
to view an edge of the Delaunay graph not only as a pair of vertices
(points in the plane), but also as a line segment, but for our application
this is not important. The \emph{Delaunay triangulation }is any completion
of a Delaunay graph to a triangulation of the plane. 

The efficiency and correctness of our algorithm depends on two key
properties of the Delaunay graph.
\begin{lem}
\label{lem:planar-property}\cite[Theorem~9.5]{ComputationalGeometry2008}
Delaunay graphs and Delaunay triangulations are planar graphs.
\end{lem}

\begin{lem}
\label{lem:closed-disk-property}\cite[Theorem~9.6 part~ii]{ComputationalGeometry2008}$\{\lambda_{i},\lambda_{j}\}$
is an edge of the Delaunay graph if and only if there is a closed
disk that contains $\lambda_{i}$ and $\lambda_{j}$ on its boundary
and does not contain any other point of $\Lambda$.
\end{lem}
We are now ready to state and prove our main result. 
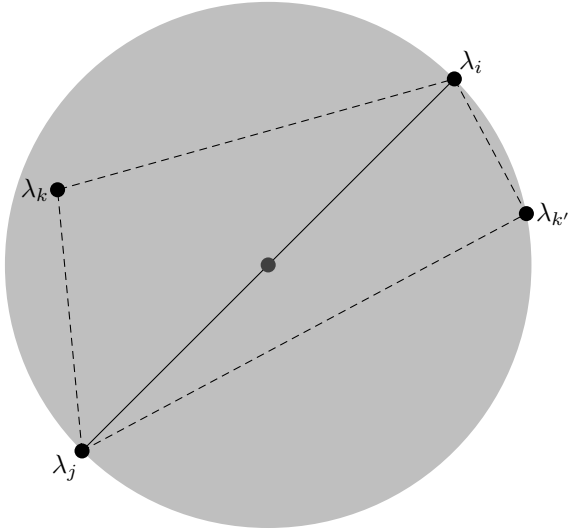
\begin{figure}
\begin{centering}
\begin{tikzpicture}
  \fill [lightgray] (4,4) circle [radius=3.5];
  \draw (6.47487373415,6.47487373415) -- (1.52512626585,1.52512626585);
  \fill[color=black] (6.47487373415,6.47487373415) circle (1mm);
  \fill[color=black] (1.52512626585,1.52512626585) circle (1mm);

  \fill[color=black] (7.43274848141,4.68281612706) circle (1mm);
  \fill[color=black] (1.2,5) circle (1mm);
  \fill[color=darkgray] (4,4)   circle (1mm);
 
  \draw (6.7,6.7) node {$\lambda_{i}$};
  \draw (1.3,1.3) node {$\lambda_{j}$};
  \draw (7.8,4.68) node {$\lambda_{k'}$};
  \draw (0.9,5) node {$\lambda_{k}$};

  \draw[densely dashed] (1.52512626585,1.52512626585) -- (7.43274848141,4.68281612706) -- (6.47487373415,6.47487373415);
  \draw[densely dashed] (1.52512626585,1.52512626585) -- (1.2,5) -- (6.47487373415,6.47487373415);
\end{tikzpicture}
\par\end{centering}
\caption{An illustration of the proof of Theorem~\ref{thm:cc-of-delaunay-graph}.
The eigenvalues $\lambda_{i}$ and $\lambda_{j}$ lie on a diameter
of the gray disk. The illustration shows both $\lambda_{k}$ in the
interior of the disk and $\lambda_{k'}$ on its boundary. The length
of the diameter is at most $\delta$ but the dashed segments are all
strictly shorter than the diameter.}
\end{figure}

\begin{thm}
\label{thm:cc-of-delaunay-graph}Let $\Lambda=\{\lambda_{1},\lambda_{2},\ldots,\lambda_{n}\}$
be a set of points in the plane, let $G(\Lambda,\delta)$ be the graph
whose vertex set is $\Lambda$ and whose edge set contains all the
pairs $\{\lambda_{i},\lambda_{j}\}$ for which the Euclidean distance
between $\lambda_{i}$ and $\lambda_{j}$ is at most $\delta$, for
some real $\delta>0$. Let $D(\Lambda)$ be the Delaunay graph of
$\Lambda$ and let $D(\Lambda,\delta)$ the subset of the graph that
contains only Delaunay edges with length at most $\delta$. We claim
that $G(\Lambda,\delta)$ and $D(\Lambda,\delta)$ have identical
connected components. 
\end{thm}
\begin{proof}
Since the edge set of $D(\Lambda,\delta)$ is a subset of the edge
set of $G(\Lambda,\delta)$, the connected components of $D(\Lambda,\delta)$
are subsets of the connected components of $G(\Lambda,\delta)$. That
is, for every $\lambda_{i}$ we have $C^{(D)}(\lambda_{i})\subseteq C^{(G)}(\lambda_{i})$.
It remains to show that $C^{(G)}(\lambda_{i})\subseteq C^{(D)}(\lambda_{i})$
also holds. We prove this claim by showing that for every edge $\{\lambda_{i},\lambda_{j}\}$
in $G(\Lambda,\delta)$ there is a path between $\lambda_{i}$ and
$\lambda_{j}$ in $D(\Lambda,\delta)$.

Assume the contrary, namely, there is an edge $\{\lambda_{i},\lambda_{j}\}$
in $G(\Lambda,\delta)$ such that there is no path in $D(\Lambda,\delta)$
connecting the vertices $\lambda_{i}$ and $\lambda_{j}$. Of all
such edges, let $\{\lambda_{i},\lambda_{j}\}$ be such that $|\lambda_{i}-\lambda_{j}|$
is the smallest (that is, the Euclidean distance between the eigenvalues
is the shortest). In particular, $\lambda_{i}$ and $\lambda_{j}$
are not connected by an \emph{edge} in the Delaunay graph, even though
the distance between them is at most $\delta$ as the edge $\{\lambda_{i},\lambda_{j}\}$
appears in $G(\Lambda,\delta)$.  Lemma~\ref{lem:closed-disk-property}
implies that every circle with $\lambda_{i}$ and $\lambda_{j}$ on
its boundary contains a third point of $\Lambda$, in the interior
or on its boundary. Consider the specific circle for which $\lambda_{i}$
and $\lambda_{j}$ lie on a diameter and let $\lambda_{k}\in\Lambda$
be a point inside that circle or on its boundary. As we have just
observed, since $\lambda_{i}$ and $\lambda_{j}$ are endpoints of
an edge in $G(\Lambda,\delta)$, the length $|\lambda_{i}-\lambda_{j}|$
is at most $\delta$. Then, both $|\lambda_{i}-\lambda_{k}|$ and
$|\lambda_{k}-\lambda_{j}|$ are smaller than $\delta$. Now, we have
two cases: (i) $\lambda_{i}$ and $\lambda_{k}$ are connected in
$D(\Lambda,\delta)$, and $\lambda_{k}$ and $\lambda_{j}$ are connected
in $D(\Lambda,\delta)$. But this forms a path in $D(\Lambda,\delta)$
between $\lambda_{i}$ and $\lambda_{j}$, which contradicts our assumption
that such a path does not exist. (ii) One of the pairs in Case~(i)
is not connected in $D(\Lambda,\delta)$: then either $\lambda_{i}$
and $\lambda_{k}$ are \emph{not }connected in $D(\Lambda,\delta)$,
or $\lambda_{k}$ and $\lambda_{j}$are \emph{not }connected in $D(\Lambda,\delta)$,
or both are not connected. Obviously, both pairs are connected in
$G(\Lambda,\delta)$ because the distances are shorter than $\delta$.
However this contradicts the fact that $\lambda_{i}$ and $\lambda_{j}$
are the pair with this property having minimum distance between them.
In either case we have a contradiction, which proves our assertion.\hfill
\end{proof}

\noindent We now prove Theorem~\ref{thm:correctness-of-sort-and-split}.
\begin{proof}
When all the eigenvalues are real, their Voronoi cells are infinite
slabs separated by vertical lines that cross the real axis half way
between adjacent eigenvalues. Therefore, all the edges of the Delaunay
triangulation connect adjacent eigenvalues. Delaunay edges longer
than $\lambda$ are prunned from $D(\Lambda,\delta)$, implying that
the sort-and-split algorithm indeed forms the connected components
of $D(\Lambda,\delta)$.\hfill
\end{proof}
Lemma~\ref{lem:planar-property} guarantees that the number of edges
in the Delaunay graph is only $O(n)$. The Delaunay graph can be easily
computed from the Voronoi diagram in $O(n)$ time, and the Voronoi
diagram itself can be computed in $O(n\log n)$ time and $O(n)$ storage~\cite[Theorem~7.10]{ComputationalGeometry2008}.
There are also randomized algorithms that compute the Delaunay triangulation
directly in $O(n\log n)$ expected time~~\cite[Theorem~9.12 and Section~9.6]{ComputationalGeometry2008}.
We can use algorithms that compute the Delaunay triangulation directly
and not the Delaunay graph because every edge that is added to the
graph to triangulate it and that remains after pruning long edges
(its length is at most $\delta$) is also an edge of $G(\Lambda,\delta)$,
so it does not modify the connected components that we compute.

\section{\label{sec:Implementation}Implementation}

We have implemented two different algorithms for the complex case,
one of them using two different arithmetic systems. All the algorithms
were implemented in C++. We implemented the algorithm that constructs
$G(\Lambda,\delta)$ explicitly and that computes its connected components
using linked-lists to represent the disjoint-set data structure. The
complexity of this implementation is $\Theta(n^{2})$. We also implemented
an algorithm that computes the Delaunay triangulation, prunes edges
longer than $\delta$ from it to form $D(\Lambda,\delta)$, and computes
the connected components of $D(\Lambda,\delta)$. The computation
of the Delaunay triangulation was done using the GAL library~\cite{CGAL-ICMS2014}\footnote{The web site of CGAL is \url{www.cgal.org}; it includes the software
and its documentation.}. CGAL allows the use of several arithmetic systems; we tested the
algorithm using three different ones, including two that are exact,
as explained later.

We used exact arithmetic to run the Delaunay triangulation because
computational-geometry algorithms can fail catastrophically when implemented
in floating point arithmetic~\cite{ClassromExamplesOfRobustnessProblems}\footnote{See also \url{http://resources.mpi-inf.mpg.de/departments/d1/projects/ClassroomExamples/}.}.
Briefly, this is caused because the algorithms compute many predicates
of the input objects and of computed geometric objects and the use
of floating-point arithmetic can easily lead to a set of binary outcomes
of the predicates that are not consistent with any input. Arithmetic
operations carried out on exact representations can be expensive and
unlike floating-point arithmetic opeartions, do not necessarily run
in constant time each. Therefore, the use of asymptotic operation
counts, such as $\Theta(n\log n)$ operations, may not have much predictive
value for actual running times. To address this, we report below on
experiments that show that the Delaunay algorithm is faster than a
naive algorithm that constructs all of $G(\Lambda,\delta)$, even
when the latter is implemented in floating-point arithmetic.

More specifically, we ran the Delaunay-based algorithm using \emph{double-precision
floating point arithmetic}, using \emph{rational arithmetic}, and
using \emph{filtered rational arithmetic}. The most informative results
are those of the filtered arithmetic, which is exact but which resorts
to the use of rational numbers only when the use of floating-point
numbers cannot guarantee the correct evaluation of a predicate. This
arithmetic system is usually almost as fast as floating-point arithmetic;
it slows down only in difficult cases. The floating-point performance
results are presented mostly in order to quantify the cost of exact
arithmetic. The pure rational results are presented mostly to demonstrate
the effectiveness of the filtered arithmetic system.

CGAL includes two implementations of algorithms that compute the 2-dimensional
Delaunay triangulation\cite{AmentaEtAl-BRIO-2003,TriangulationsInCGAL,Devillers-Incremental-1998}.
One algorithm is an incremental algorithm that inserts points into
the triangulation in a biased randomized order~\cite{AmentaEtAl-BRIO-2003}.
An insertion of a vertex with degree $d$ costs $\Theta(d)$. The
expected running time of this algorithm is $O(n\log n)$, but the
worst-case running time is $O(n^{2})$. When the set of points is
not known in advance, a different algorithm that maintains a Delaunay
hierarchy~\cite{Devillers-Incremental-1998} often runs faster, but
this is not the case in our application. We tested the Delaunay-hierarchy
variant and it was indeed a bit slower. The asymptotic worst-case
and expected running times of this algorithm are similar to (or worse,
for some insertion orders) those of the random-insertion-order algorithm.
Curiously, CGAL does not include a deterministic worst-case $O(n\log n)$
Delaunay-triangulation algorithm, even though such algorithms are
older than the incremental algorithms~\cite{ComputationalGeometry2008,FortuneVoronoiDelaunay2018};
it appears that they are usually slower in practice than the expected-case
$O(n\log n)$ algorithms. 

The naive algorithms do not require an exact-arithmetic implementation,
since distance computations in the plane are accurate (see~\cite[Sections~3.1 and~3.6]{HighamFoM}),
which means that the computed edge set of $G(\Lambda,\delta)$ will
include all the edges whose exact length is below $\delta-\epsilon$
and will exclude all the edges whose length is above $\delta+\epsilon$,
both for some $\epsilon$ much smaller than $\delta$. We also use
floating-point arithmetic to prune the Delaunay triangulation, for
the same reason.

\subsection{Parallelism}

The algorithms that we propose can be parallelized, but by reducing
the total work to $\Theta(n\log n)$, we essentially eliminate the
need to parallelize this part of the Davies-Higham method.

The algorithm for the real case can be easily and effectively parallelized,
because there are effective parallel algorithms for both sorting and
parallel prefix~\cite{LeightonParallel}. The algorithm for the complex
case requires a parallel two-dimensional Delaunay triangulation. Several
such algorithms have been developed~\cite{BlellochEtAlParallelDelaunay,ParallelDC2DDelaunay,TIPP},
but unfortunately, none of them have been implemented in CGAL or in
any other robust library.

However, given that our algorithm are designed to be used in an $\Omega(n^{3})$
method whose critical path has length $\Omega(n)$ (the Schur decomposition),
the $O(n\log n)$ cost of our algorithm is unlikely to create a significant
Amdahl-type bottleneck even if it remains sequential. 

\section{\label{sec:Experimental-Results}Experimental Results}

\begin{figure*}
\includegraphics[clip,width=0.48\textwidth]{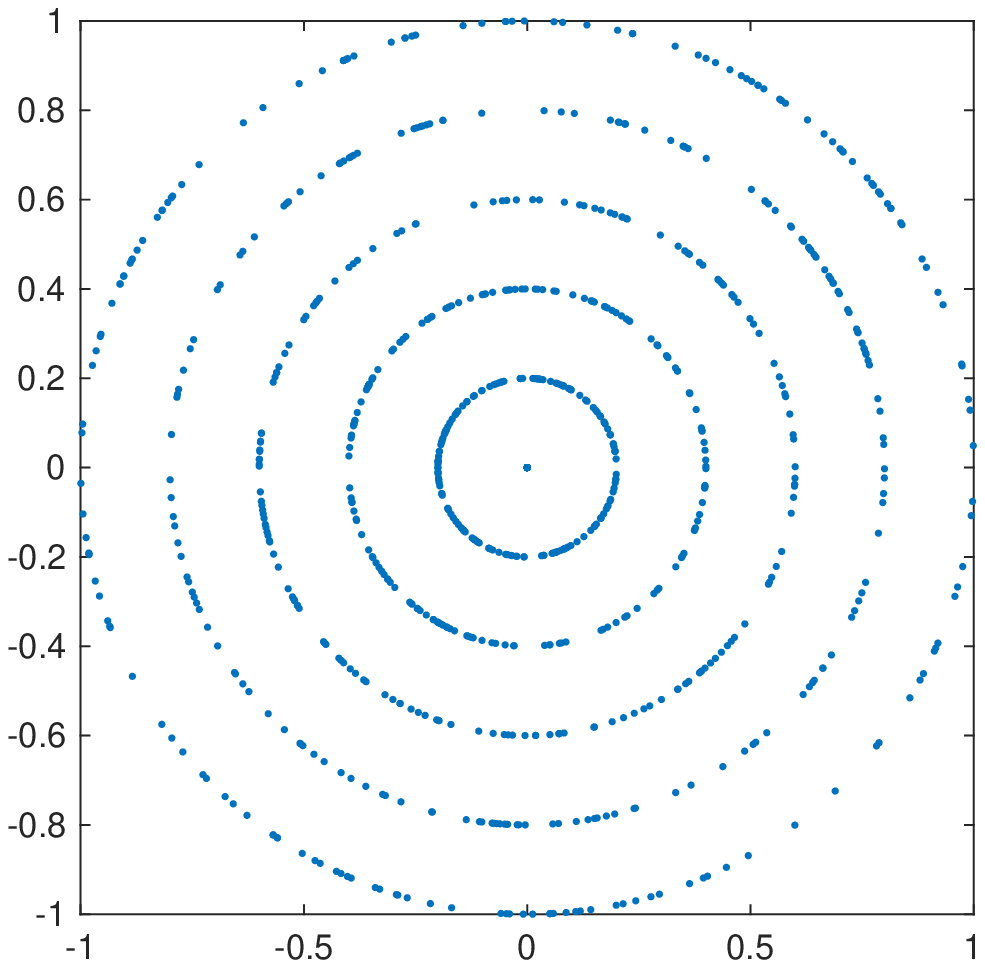}\hfill{}\includegraphics[clip,width=0.48\textwidth]{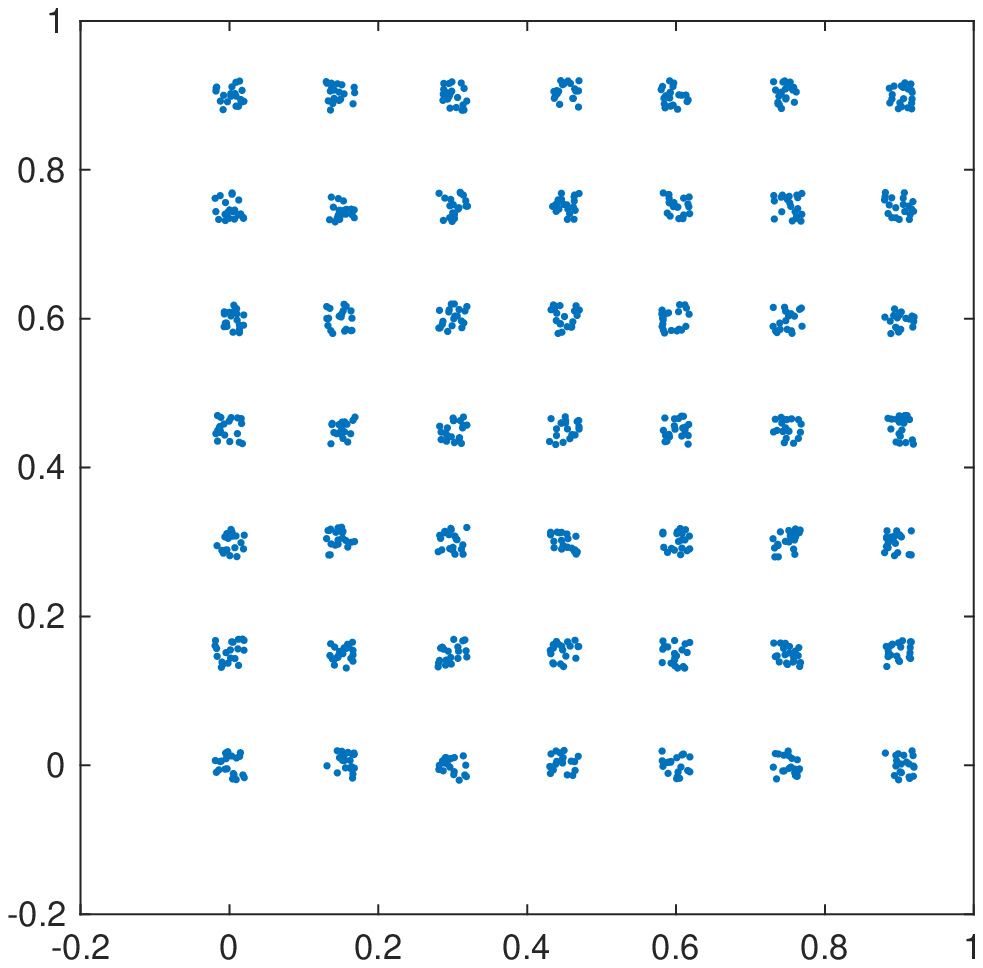}

\caption{\label{fig:eigenvalue-distributions}Examples of eigenvalue distributions
that we used for testing. The plot on the left shows $999$ eigenvalues
placed randomly on one of 5 circles that are spaced $0.2$ apart,
as well as one eigenvalue at the origin. The eigenvalues are distributed
evenly among the circles and the angular position of each eigenvalue
is random with uniform distribution. The plot on the right shows $1000$
eigenvalues placed randomly and uniformly in one of several squares
with a side of length $0.04$ and with centers spaced every $0.15$.
Again, the eigenvalues are distributed among the squares evenly.}

\end{figure*}
\begin{figure*}
\includegraphics[width=0.48\textwidth]{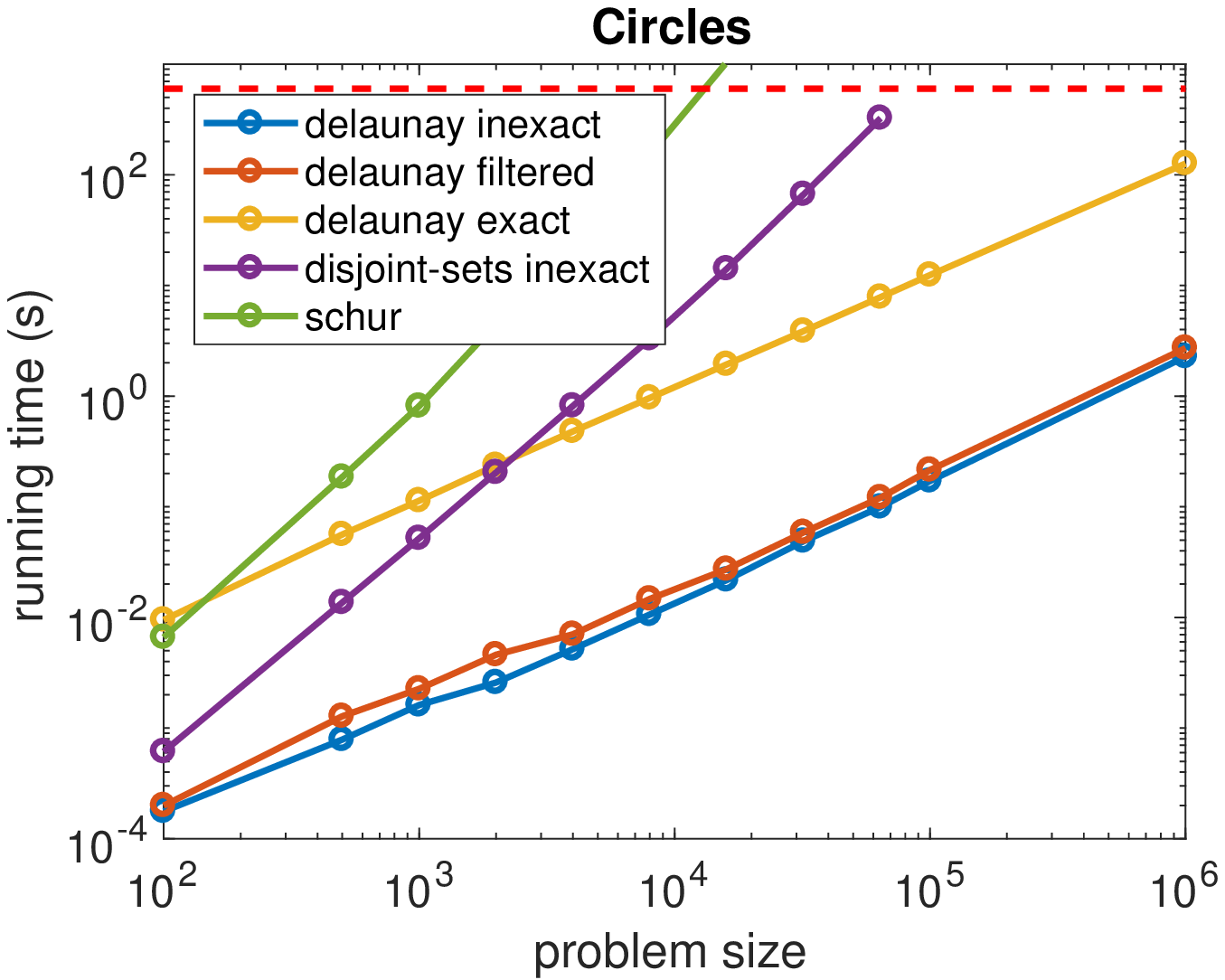}\hfill{}\includegraphics[width=0.48\textwidth]{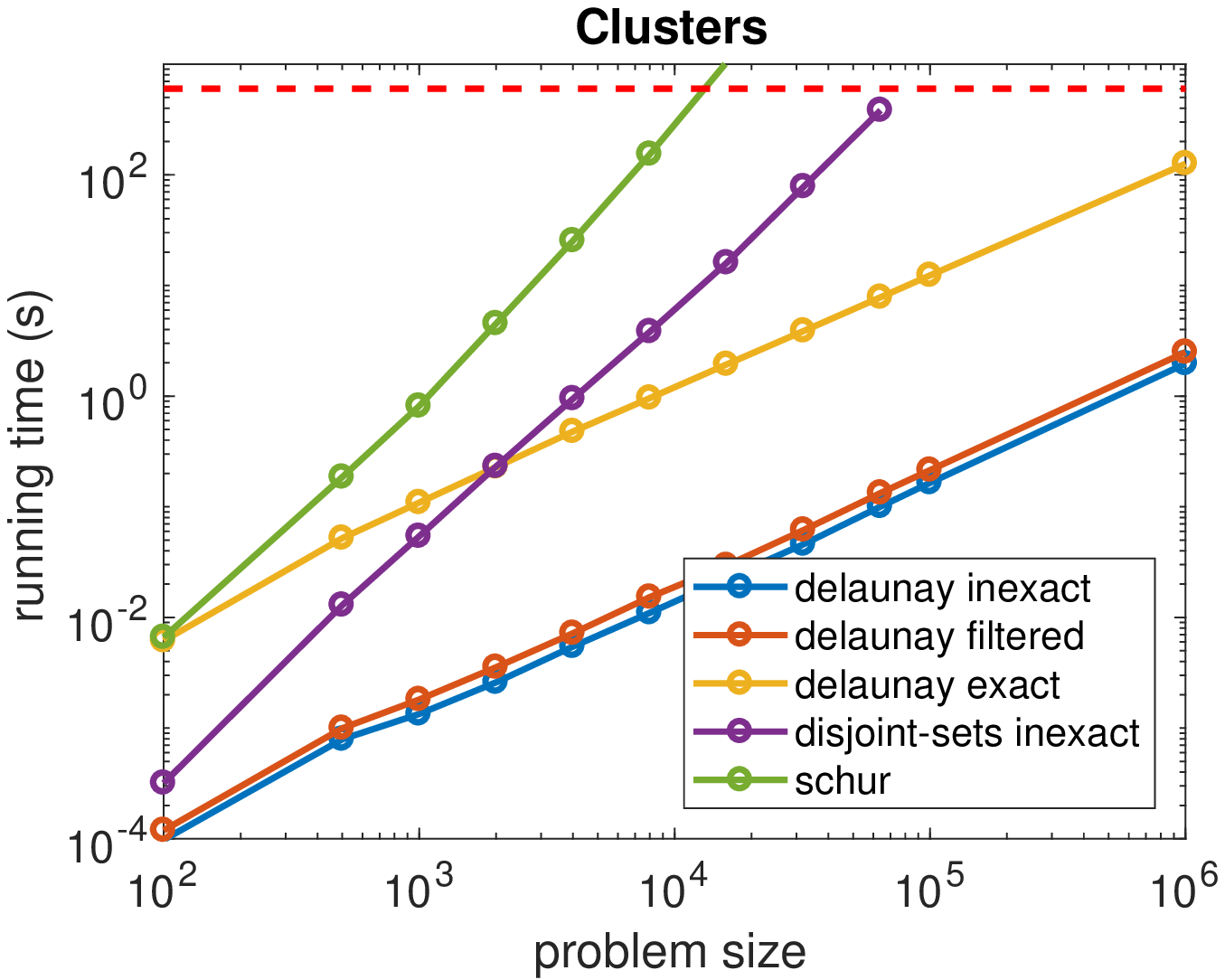}

\caption{\label{fig:results-loglog}The running time of the algorithms on two
different distributions of eigenvalues, as a function of the problem
size $n$. The distribution used to produce the graph on the left
is exactly the distribution shown on the left in Figure~\ref{fig:eigenvalue-distributions}.
The distribution used to produce the graph on the right is similar
to the distribution shown on the right in Figure~\ref{fig:eigenvalue-distributions},
but each cluster was distributed uniformly in a square with side length
$0.02$. The running times did not change much when we modified the
side of these squares to $0.15$ (so that the eigenvalues are distributed
almost uniformly in the unit square) and to $2\times10^{-10}$, a
very tight clustering. We stopped very slow runs after 10 minutes;
the dashed red line shows this limit.}
\end{figure*}
We conducted experiments to assess the running times of the algorithms.
We compiled the codes using the Microsoft C++ compiler version~19
using the O2 optimization level and we used version~4.13 of GCAL.
We also include, for reference, the running times of the $\Theta(n^{3})$
Schur decomposition in Python's \texttt{scipy.linalg} package, which
uses LAPACK and an optimized version of the BLAS to compute the decomposition.
This decomposition is the first step in the Davies-Higham method.

We ran all the experiments on a computer with a quad-core 3.5~GHz
i5-4690K processor and 16~GB of 800~MHz DDR3 DRAM running Windows
10. Our codes are single threaded, so they used only one core. The
Schur-decomposition runs used all the cores.

We evaluated the algorithms using two different distributions of eigenvalues,
illustrated in Figure~\ref{fig:eigenvalue-distributions}. One distribution
includes an eigenvalue at the origin and the rest are placed on concentric
circles with radii that differ by more than $\delta$, so that no
cluster spans more than one circle. Each circle contains approximately
the same number of eigenvalues and the location of each eigenvalue
on its circle is random and uniform. More specifically, we used $\delta=0.1$,
as recommended by Davies and Higham, and radii separation of $0.2$.
The other distribution splits the eigenvalues evenly among squares
whose centers are more than $\delta$ apart. The eigenvalues in each
square are distributed uniformly in the square. We tested this distribution
with $\delta=0.1$ and squares whose centers are $0.15$ apart, and
with sides of $2\times10^{-10}$, $0.02$, or $0.15$. In the first
two cases (sides of $2\times10^{-10}$ and $0.02$) the eigenvalues
in each square form a single cluster, separate from those of other
squares. When the squares have sides of length $0.15$, clusters often
span more than one square (the eigenvalues are distributed approximately
uniformly in the unit square).

The results show that the Delaunay-based algorithm are much faster
than the naive algorithm. The different slopes on the log-log scale
indicate that the algorithms run in approximately polynomial times
but with different polynomial degrees. The results also show that
the overhead of exact arithmetic, when using the filtered implementation,
is minor. The overhead of naive rational arithmetic is considerable;
it is more than 60 times slower than floating-point arithmetic and
about 50 times slower than the filtered rational arithmetic.

\begin{figure*}
\includegraphics[width=0.48\textwidth]{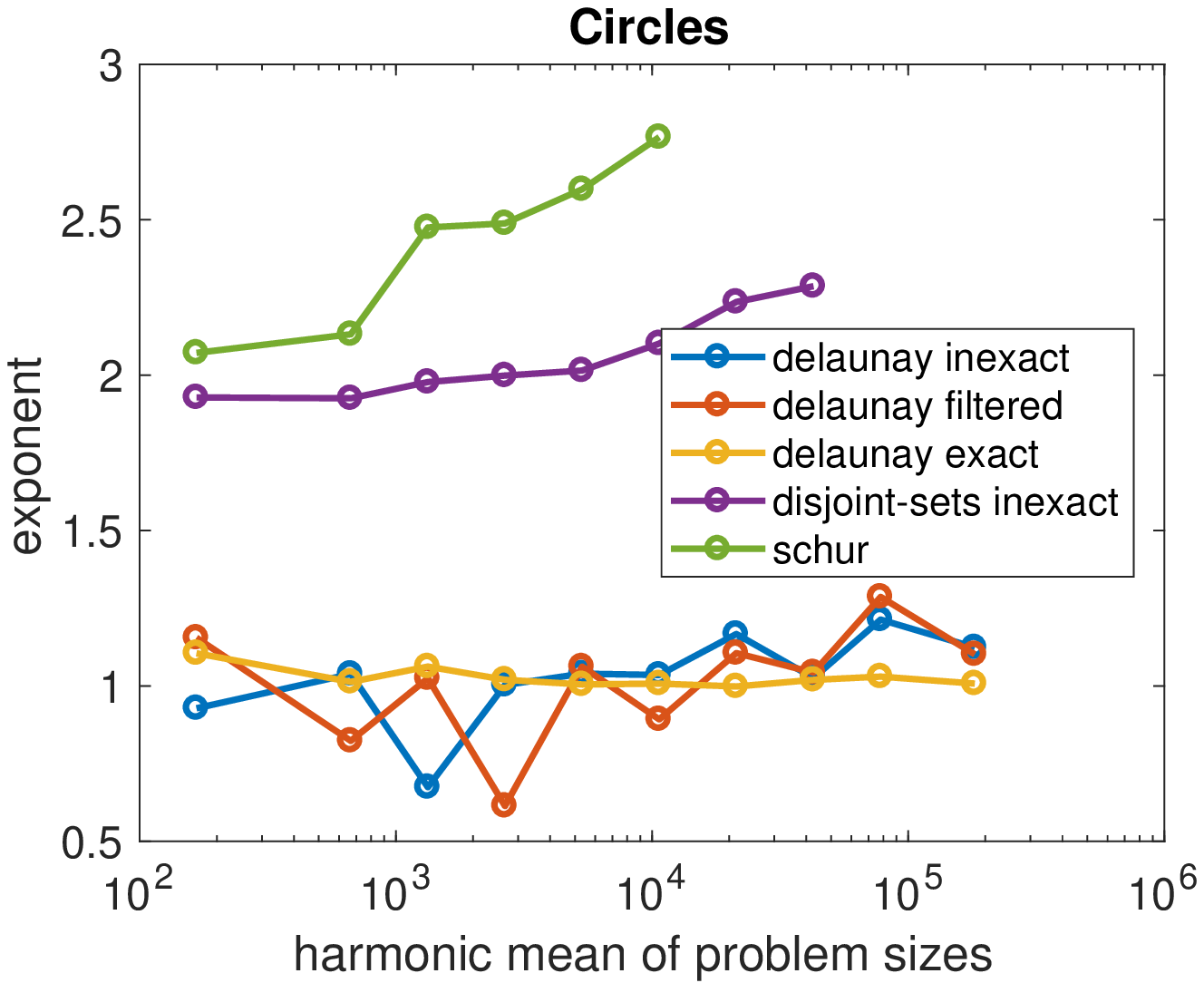}\hfill{}\includegraphics[width=0.48\textwidth]{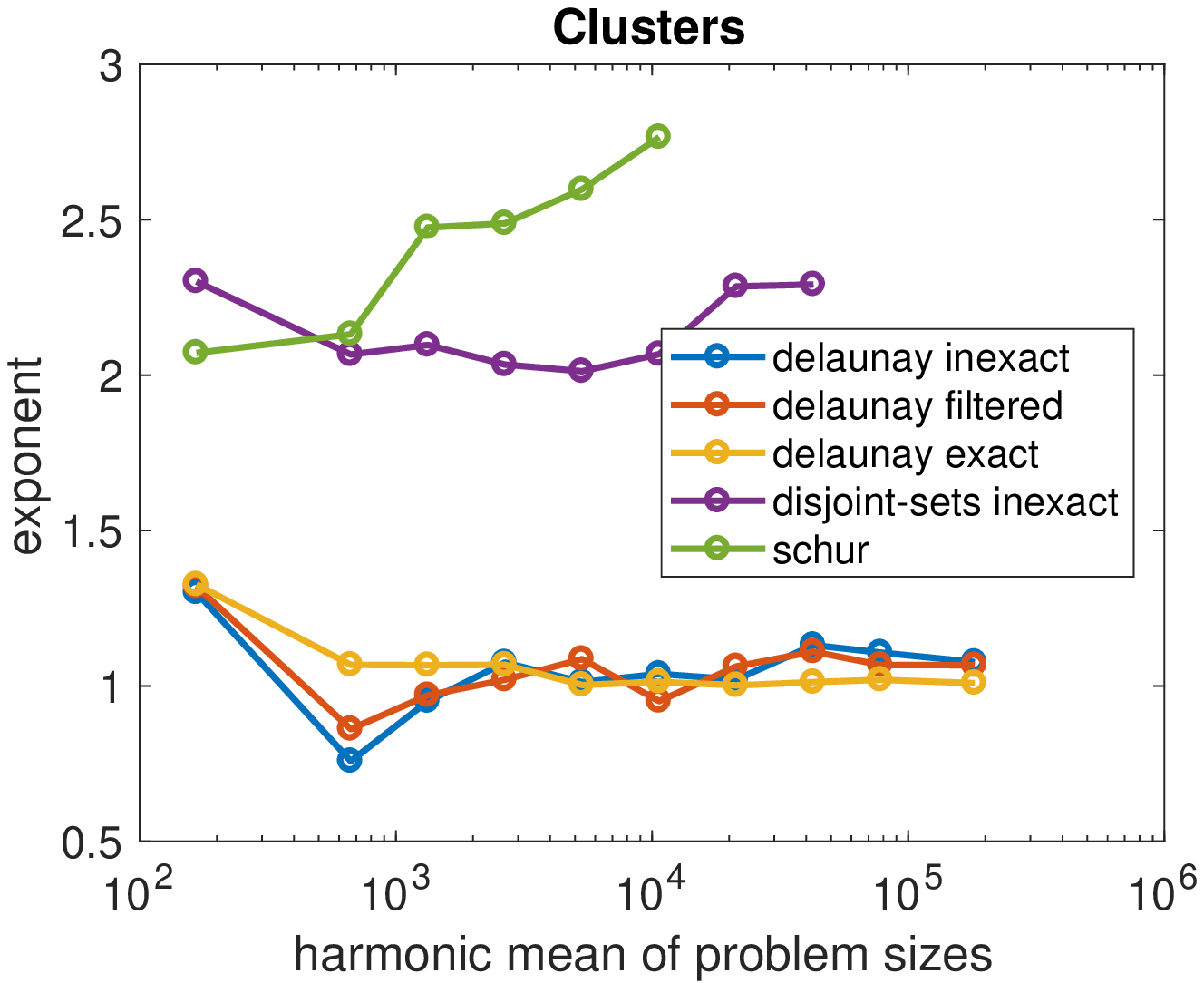}

\caption{\label{fig:results-exponents}Estimates of the polynomial degree of
the running times. For each pair of successive problem sizes $n_{1}$
and $n_{1}$ that resulted in running times of $T(n_{1})$ and $T(n_{2})$
using a particular algorithm, these graphs plot $\log(T(n_{2})/T(n_{2}))/\log(n_{2}/n_{1})$
against the harmonic mean of $n_{1}$and $n_{2}$. }
\end{figure*}
Figure~\ref{fig:results-exponents} estimates the degree of the polynomial
running times. For each algorithm and for each pair of running times
$T(n_{1})$ and $T(n_{2})$ on problems of sizes $n_{1}$ and $n_{1}$,
the graphs show 
\[
\frac{\log\left(T(n_{2})/T(n_{2})\right)}{\log\left(n_{2}/n_{1}\right)}
\]
as a function of the harmonic mean of $n_{1}$ and $n_{2}$,
\[
\left(\frac{n_{1}^{-1}+n_{2}^{-1}}{2}\right)^{-1}\;.
\]
In particular, if $T(n)=n^{d}$, then $\log(T(n_{2})/T(n_{2}))/\log(n_{2}/n_{1})=d$.
The results show that the running times of the Delaunay algorithm
are approximately linear in the problem size (the exponent is close
to $1$) whereas the running times of the naive algorithms are worse
than quadratic. The results also show that the running times of the
rational arithmetic implementation are smoother than those of the
floating-point and filtered implementations. We believe that the worse-than
quadratic behavior of the naive implementation is due to increasing
cache-miss rates, but we have not tested this hypothesis directly.

\subsection{Quadratic Behavior in the Delaunay-Based Algorithm.}

A variant of the concentric-circles eigenvalue distribution induced
quadratic running times in the Delaunay algorithm. In that variant,
the zero eigenvalue had multiplicity $n/6$. Each of the other circles,
and in particular the circle with radius $0.2$, also had about $n/6$
of the eigenvalues. This implies that the Voronoi cell of the origin
is a polygon with approximately $n/6$ edges, which implies that the
degree of the origin in the Delaunay triangulation is also about $n/6$.
This implies that the cost of inserting this point, in the incremental
algorithms, is $\Theta(n)$. This cost recurs for each instance of
the eigenvalue, bringing the total cost to $\Theta(n^{2})$.

There are two ways to address this difficulty; both work well. One
solution is to eliminate (exactly) multiple eigenvalues by sorting
them (e.g, lexicographically). Only one representative of each multiple
eigenvalue need be included in the clustering algorithm; the rest
are automatically placed in the same cluster. The total cost of this
approach is $O(n\log n)$. A hash table can reduce the cost even further.

The other approach is to perturb eigenvalues, say by $\|A\|_{1}\sqrt{\epsilon}$,
where $\epsilon$ is the machine epsilon (unit roundoff) of the arithmetic
in which the eigenvalues have been computed. This may modify the clusters
slightly, but since the Davies-Higham algorithm requires a very large
separation ($0.1$), the difference is unlikely to modify the stability
of the overall algorithm.

\section{\label{sec:Conclusions}Conclusions}

We have presented an efficient algorithm to cluster eigenvalues for
the Davies-Higham method for computing matrix functions. The algorithm
is based on a sophisticated computational-geometry building block.
Its implementation exploits CGAL, a computational-geometry software
library, and uses a low-overhead exact arithmetic. The new algorithm
outperforms the previous algorithm, proposed by Davies and Higham,
by large margins. 

\paragraph*{Acknowledgments}

We thank Olivier Devillers for clarifying the behavior of Delaunay
triangulations in CGAL. This research was support in part by grants
825/15, 863/15, 965/15, 1736/19, and 1919/19 from the Israel Science
Foundation (funded by the Israel Academy of Sciences and Humanities),
by the Blavatnik Computer Science Research Fund, and by grants from
Yandex and from Facebook.

We thank the anonymous referees for useful feedback and suggestions.

\balance

\bibliographystyle{plain}
\bibliography{functions-of-matrices,PolynomialEvaluation}

\section*{Appendix A}

Davies and Higham imply that the connected components of $G(\Lambda,\delta)$
are equivalent to the a partition that satisfies the following two
conditions, but this is not the case.
\begin{defn}
Given some real $\delta>0$, a \emph{$\delta$-admissible} partitioning
of a set of complex numbers $\Lambda=\{\lambda_{1},\lambda_{2},\ldots,\lambda_{n}\}$
(possibly with repetitions) into clusters (subsets) $C_{1},\ldots,C_{k}$
satisfies the following two conditions.
\end{defn}
\begin{enumerate}
\item Separation between clusters: $\min\{|\lambda_{i}-\lambda_{j}|:\lambda_{i}\in C_{p},\lambda_{j}\in C_{q},p\neq q\}>\delta$.
\item Separation within clusters: if $|C_{p}|>1$, then for every $\lambda_{i}\in C_{p}$
there is a $\lambda_{j}\in C_{p}$, $i\neq j$, such that $|\lambda_{i}-\lambda_{j}|\leq\delta$.
\end{enumerate}
Partitioning into connected components is always admissible.
\begin{thm}
The connected components of $G(\Lambda,\delta)$ form an admissible
partitioning of $\Lambda$.
\end{thm}
\begin{proof}
Let $C_{1},\ldots,C_{k}$ be the connected components of $G(\Lambda,\delta)$.
Admissibility criterion~1 is satisfied because if for some $\lambda_{i}\in C_{p},\lambda_{j}\in C_{q}$
we have $|\lambda_{i}-\lambda_{j}|\leq\delta$, then $\{\lambda_{i},\lambda_{j}\}$
is an edge of $G(\Lambda,\delta)$ so the vertices must be in the
same connected component, implying $p=q$. The second criterion is
also satisfied because if $C_{p}$ is a non-singleton connected component,
then every vertex $\lambda_{i}\in C_{p}$ in the component must have
a neighbor $\lambda_{j}\in C_{p}$, and by the neighborhood relationship
we have $|\lambda_{i}-\lambda_{j}|\leq\delta$.
\end{proof}
However, not every admissible partitioning is a partitioning into
connected components. 
\begin{example}
Let $0<\delta<1/2$ and let $\Lambda$ consist of $\{1,1+\delta,2,2+\delta\}$.
The edge set of $G(\Lambda,\delta)$ consists of pairs of the form
$\{1,1+\delta\}$ and $\{2,2+\delta\}$, and these are also the two
connected components of the graph. This is also an admissible partitioning,
but a trivial partitioning $C_{1}=\Lambda$ is also admissible. The
separation-between-clusters criterion is satisfied trivially; the
minimization is over an empty set. The separation-within-clusters
criterion is also satisfied, because every vertex in $C_{1}$ is close
to some other vertex in $C_{1}$, its neighbor in $G(\Lambda,\delta)$.
\end{example}
The admissibility criteria do guard the numerical stability, but they
allow larger-than-necessary clusters, which increase the computational
complexity of the Davies-Higham method.
\end{document}